\documentclass{article}
\usepackage[utf8]{inputenc}
\usepackage{amsmath}
\usepackage{amsthm}
\usepackage{amssymb}
\usepackage{hyperref}
\usepackage{enumerate}
\usepackage{graphicx}

\newtheorem{theorem}{Theorem}
\newtheorem{lemma}{Lemma}

\newtheorem{corollary}{Corollary}

\newtheorem{claimx}{Claim}

\theoremstyle{definition}
\newtheorem{definition}{Definition}

\newcommand{\Z}{\mathbb{Z}}

\newcommand{\defproblem}[3]{
  \vspace{2mm}
\noindent\fbox{
  \begin{minipage}{0.95\textwidth}
  #1 \\
  {\bf{Input:}} #2  \\
  {\bf{Question:}} #3
  \end{minipage}
  }
  \vspace{2mm}
}
\newcommand{\deftproblem}[3]{
  \vspace{2mm}
\noindent\fbox{
  \begin{minipage}{0.95\textwidth}
  #1 \\
  {\bf{Input:}} #2  \\
  {\bf{Task:}} #3
  \end{minipage}
  }%
}

\author{Ond\v rej Such\' y\\[\smallskipamount]
Department of Theoretical Computer Science,\\
Faculty of Information Technology,\\
Czech Technical University in Prague,
Prague, Czech Republic\\
\url{ondrej.suchy@fit.cvut.cz}}
\title{On Directed Steiner Trees with Multiple Roots}
\date{\today}

\begin{document}
\maketitle
\begin{abstract}
We introduce a new Steiner-type problem for directed graphs named \textsc{$q$-Root Steiner Tree}. Here one is given a directed graph $G=(V,A)$ and two subsets of its vertices, $R$ of size $q$ and $T$, and the task is to find a minimum size subgraph of $G$ that contains a path from each vertex of $R$ to each vertex of $T$. The special case of this problem with $q=1$ is the well known \textsc{Directed Steiner Tree} problem, while the special case with $T=R$ is the \textsc{Strongly Connected Steiner Subgraph} problem.

We first show that the problem is W[1]-hard with respect to $|T|$ for any $q \ge 2$. Then we restrict ourselves to instances with $R \subseteq T$. Generalizing the methods of Feldman and Ruhl [SIAM J. Comput. 2006], we present an algorithm for this restriction with running time $O(2^{2q+4|T|}\cdot n^{2q+O(1)})$, i.e., this restriction is FPT with respect to $|T|$ for any constant $q$.
We further show that we can, without significantly affecting the achievable running time, loosen the restriction to only requiring that in the solution there are a vertex $v$ and a path from each vertex of $R$ to $v$ and from $v$ to each vertex of~$T$.

Finally, we use the methods of Chitnis et al. [SODA 2014] to show that the restricted version can be solved in planar graphs in $O(2^{O(q \log q+|T|\log q)}\cdot n^{O(\sqrt{q})})$ time. 
\end{abstract}

\section{Introduction}
Steiner type problems are one of the most fundamental problems in the network design. In general words the task is to connect a given set of points at the minimum cost. The study of these problems in graphs was initiated independently by Hakimi~\cite{Hakimi71} and Levin~\cite{Levin71}. In the classic \textsc{Steiner Tree} one is given a (weighted) undirected graph $G=(V,E)$ and a set $T$ of its vertices (\emph{terminals}) and the task is to find a minimum cost connected subgraph containing all the terminals. 

In directed graphs, the notion of connectivity is more complicated. The notion which turns out to be the closest to the undirected \textsc{Steiner Tree} is that of \textsc{Directed Steiner Tree (DST)}, where one is given a (weighted) directed graph $G=(V,A)$, a set $T$ of terminals, and additionally a root vertex $r$ and the task is to find a minimum weight subgraph that provides a path from $r$ to each of~$T$. Another natural option is, given a digraph $G=(V,A)$ and a set $T$ of terminals, to search for a minimum cost subgraph that provides a path between each pair of terminals in both directions. This is problem is called \textsc{Strongly Connected Steiner Subgraph (SCSS)}. The most general problem allows to prescribe the demanded connection between the terminals. Namely, in \textsc{Directed Steiner Network (DSN)} one is given a digraph $G=(V,A)$ and a set of $q$ pairs of vertices $\{(s_1,t_1), \ldots, (s_q,t_q)\}$ and is asked to find a minimum weight subgraph $H$ of $G$ that contains a directed path from $s_i$ to $t_i$ for every $i$.

Obviously, DSN is a generalization of both DST and SCSS. In this paper we consider a special case of DSN, which is still a very natural generalization of both DST and SCSS, namely the following problem:\\
\defproblem{\textsc{$q$-Root Steiner Tree ($q$-RST)}}
{A directed graph $G=(V,A)$, two subsets of its vertices $R,T \subseteq V$ with $|R| = q$, and a positive integer $k$.}
{Is there a set $S \subseteq V$ of size at most $k$ such that in $G[R \cup S \cup T]$ there is a directed path from $r$ to $t$ for every $r \in R $ and every $t \in T$?}

If $q=1$, then $q$-RST problem is equal to (unweighted) DST. On the other hand, if we let $T=R$, then the problem is equivalent to (unweighted) SCSS on the terminal set $T$. We study the problem from a multivariate perspective, examining the influence of various parameters on the complexity of the problem. We focus on the following parameters: number of roots $q = |R|$, number of terminals $|T|$, and to a limited extent also to the budget $k$.
Thorough the paper we denote $n=|V|$ and $m=|A|$. Before we present our results, let us summarize what is known about the problems.
\paragraph*{Known Results}

\textsc{Steiner Tree} is NP-hard~\cite{GareyJ79} and remains so even in very restricted planar cases~\cite{GareyJ77}. As the NP-hardness can be easily transferred also to DST and SCSS, the problems were studied from approximation perspective. However, in general terms, the problems are also hard to approximate. The best known approximation factor for DST and SCSS is~$O(|T|^{\epsilon})$ for any fixed~$\epsilon>0$~\cite{CCCD+99}. On the other hand, the problems cannot be approximated to within a factor of $O(\log^{2-\epsilon} n)$ for any $\epsilon>0$, unless NP has quasi-polynomial time Las Vegas algorithms~\cite{HalperinKKSW07}. For the most general DSN problem the best known ratio is $n^{2/3+\epsilon}$ for any $\epsilon>0$ and the problem cannot be approximated to within $O(2^{\log^{1-\epsilon} n})$ for any $\epsilon>0$, unless NP has quasi-polynomial time algorithms~\cite{Berman13}. We refer to surveys, e.g., \cite{KN07}, for more information on the numerous polynomial-time approximation results for Steiner-type problems.

From the perspective of parameterized algorithms~\cite{CyganFKLMPPS15,DowneyF13} the problems are mostly studied with respect to the number of terminals. It follows from the classical result of Dreyfus and Wagner~\cite{DW72} (independently found by Levin~\cite{Levin71}), that \textsc{Steiner Tree} and also DST can be solved in $O(3^{|T|}\cdot n^{O(1)})$ time. The algorithm was subsequently improved~~\cite{EricksonMV87,FKMRRW07,BHKK07} with the latest algorithm of Nederlof~\cite{Nederlof13} achieving $O(2^{|T|}\cdot n^{O(1)})$ time and polynomial space complexity.

For the SCSS and DSN with $q$ terminals and $q$ terminal pairs, Feldman and Ruhl~\cite{FeldmanR06} showed that the problems can be solved roughly in $O(n^{2q})$ and $O(n^{4q})$ time, respectively. We cannot expect fixed parameter tractability for these problems, since the problems are W[1]-hard with respect to this parameter~\cite{GuoNS11} (and even with respect to the total size of the sought graph). In fact, unless the Exponential Time Hypothesis (ETH)~\cite{ImpagliazzoP01} fails, SCSS cannot be solved in $f(q)n^{o(q/\log q)}$ time on general graphs and DSN cannot be solved in $f(q)n^{o(q)}$ time even on planar DAGs~\cite{ChitnisHM14}. Chitnis et al.~\cite{ChitnisHM14} also showed that on planar graphs SCSS can be solved within $2^{O(q \log q)}n^{O(\sqrt{q})}$ time, but it is still W[1]-hard and cannot be solved within $f(q)n^{o(\sqrt{q})}$ time, unless ETH fails. 

With respect to the less studied parameter ``number of nonterminals in the solution'', representing one possible measure of the solution size, all the problems are on general graphs W[2]-hard by an easy reduction from \textsc{Set Cover} (see, e.g., Guo et al.~\cite{GuoNS11}). On planar graphs, only DST was studied with respect to this parameter, achieving fixed parameter tractability~\cite{JonesLRSS13}.

\paragraph*{Our Contribution}

In this paper our aim is to generalize the positive results for DST and SCSS also to $q$-RST. Unfortunately, as our first result, we show that $q$-RST is still too general to achieve this goal. Namely, we show that for any constant $q \ge 2$ the $q$-RST is W[1]-hard with respect to $|T|$ even on directed acyclic graphs and cannot be solved within $f(|T|)n^{o(|T|/\log |T|)}$ time, unless ETH fails.
In fact the same results hold even if we replace $|T|$ by $(k+|T|)$, the total number of vertices in the resulting subgraph (minus $q$).

Then, we restrict the problem further to its special case by requiring $R \subseteq T$. In fact, for better readability we require the solution to provide a path from each $r \in R$ to each vertex $t \in R\cup T$ and assume $T \cap R = \emptyset$. We call the resulting problem \textsc{$q$-Root Steiner Tree with Pedestal} ($q$-RST-P). Observe that it still generalizes DST as well as SCSS.

We show that we can generalize the algorithm of Feldman and Ruhl~\cite{FeldmanR06} for SCSS to $q$-RST-P, using an algorithm for DST as a subroutine. The running time of our algorithm is $O(2^{2q+4|T|}\cdot n^{2q+O(1)})$, i.e., the problem is FPT with respect to $|T|$ for any constant $q$ and the exponent of the polynomial depends linearly on $q$. The lower bounds for SCSS indicate that this dependency on $q$ is optimal. In fact if $T=\emptyset$, then our algorithm is exactly the algorithm of Feldman and Ruhl, while if $q=1$, the algorithm  boils down to a single call to the DST subroutine.

The algorithm of Feldman and Ruhl is based on a token game, where the tokens trace the path required in the solution. The solution of the instance is then represented by a sequence of moves of the tokens between two specified configurations. We first enrich the game by introducing new tokens that trace the path to vertices of $T$ while using the original tokens to trace paths between the vertices in $R$. We call this game cautious. 

We then show that the solutions can be represented by move sequences with further interesting properties. these allow us to group the moves and reduce the number of intermediate configurations. The resulting game, which we call accelerated, has moves very similar to the original game of Feldman and Ruhl, but each move is now equipped by a subset of vertices of $T$ that is also reached in this move. We use this similarity for further result in our paper.

The crucial property of the problem that allows us to come up with the algorithm is that there is a vertex such that every path required by the solution can be dragged through this specific vertex (allowing the vertices to repeat on the path). To illustrate this, we introduce another variant of the problem \textsc{$q$-Root Steiner Tree with Trunk} ($q$-RST-T), which is the same as $q$-RST, but the solution is further required to contain a vertex which has a path from each vertex in $R$ and to each vertex in $T$.
We show that this problem can be solved in similar running time as $q$-RST-P, namely $O(2^{2q+4|T|}\cdot n^{3q+O(1)})$. Qualitatively similar running time can be also achieved if the special vertex provides all but a constant number of the paths required by the problem. 

We further generalize the result of Chitnis et al.~\cite{ChitnisHM14} giving the improved algorithm for SCSS in planar graphs to obtain an algorithm for $q$-RST-P in planar graphs with running time $O(2^{O(q \log q+|T|\log q)}\cdot n^{O(\sqrt{q})})$.

While the hardness result applies to the decision variant, the algorithms directly apply to the (cardinality) optimization case. Moreover, it is straightforward to generalize them to the case of vertex weights (we might want to use different, more suitable, DST algorithm as a subroutine, based on the actual range of the weights). In order to use arc weights, one just has to subdivide each arc and give the weight of the arc to the newly created vertex. Thus our algorithms also apply to vertex weighted and arc weighted variants of the problems.
Nevertheless, for ease of presentation, we formulate all our results only for the cardinality case.

\paragraph*{Organization of the paper}
In Section~\ref{sec:hardness} we present the hardness result for the unrestricted version of $q$-RST. Section~\ref{sec:pedestal} describes the games and the algorithm for $q$-RST-P. This is generalized to $q$-RST-T in Section~\ref{sec:trunk}. The improved algorithm for $q$-RST-P in planar graphs is contained in Section~\ref{sec:planar}. We conclude the paper with outlook in Section~\ref{sec:concl}. 

\section{Unrestricted case}\label{sec:hardness}
In this section we show that the unrestricted variant of $q$-RST is W[1]-hard with respect to $|T|$ even on directed acyclic graphs. Let us start with the case $q=2$.

\begin{theorem}
\label{thm:hardness}
2-RST is W[1]-hard with respect to $|T|$ even on directed acyclic graphs. Moreover, there is no algorithm for 2-RST on directed acyclic graphs running in $f(|T|)n^{o\left(\frac{|T|}{\log |T|}\right)}$ time, unless ETH fails.
\end{theorem}

The rest of this section is devoted to the proof of this theorem. 
Our starting point are known results for the following problem.\\
\defproblem{\textsc{Partitioned Subgraph Isomorphism (PSI)}}
{Undirected graphs $H=(V_H,E_H)$ and $G=(V_G,E_G)$ and a coloring function $col:V_H\rightarrow V_G$.}
{Is there an injection $\phi:V_G\rightarrow V_H$ such that for every $i\in V_G$, $col(\phi(i))=i$ and for every $\{i,j\}\in E_G$, $\{\phi(i),\phi(j)\}\in E_H$?}

PSI is known to be W[1]-hard~\cite{Pietrzak03}. For the second part of the theorem we need the following lemma by Marx~\cite{Marx07}.

\begin{lemma}[{\cite[Corollary 6.3]{Marx07}}]\label{lem:psi hardness}
{\sc Partitioned Subgraph Isomorphism} cannot be solved in $f(k)V_H^{o(\frac{k} {\log k})}$ time, where $f$ is an arbitrary function and $k=\vert E_G\vert$ is the number of edges in the smaller graph $G$, unless {ETH} fails.
\end{lemma}

We provide a parameterized reduction from PSI parameterized by $|E_G|$ to 2-RST parameterized by $|T|$. 

Let $(H=(V_H,E_H),G=(V_G,E_G),col)$ be an instance of PSI. To simplify the description let us assume that there are some strict linear orders $<$ on the vertices in $V_H$ and in $V_G$ such that for every $u$ and $v$ with $u < v$ we have $col(u)<col(v)$. 
We also assume that for every edge $\{u,v\} \in E_h$ we have $\{col(u),col(v)\} \in E_G$ as we can delete the edges not satisfying this condition without affecting the answer to the instance. Finally, since the problem can be solved for each connected component separately, we assume that $G$ is connected and has at least one edge.

We start by constructing the directed graph $G'=(V',A')$ (see also Fig.~\ref{fig:reduction}).
We let $V'= R \cup V_H \cup E' \cup F \cup T$, where
\begin{align*}
 R &= \{r_V,r_E\},\\
 E'&= \{a_{u,v} \mid \{u,v\} \in E_H, u < v\},\\
 F &= \{b_{u,v}, b_{v,u} \mid \{u,v\} \in E\}\text{, and}\\
 T &= \{t_{i,j}, t_{j,i} \mid \{i,j\} \in E_G, i < j\}. 
\end{align*}
The set of arcs is constructed as follows. We add arcs from $r_V$ to all vertices in $V_H$ and from $r_E$ to all vertices in $E'$.
For every edge $\{u,v\}$ where $u < v$, we add the following set of arcs:
\begin{itemize}
 \item an arc from $a_{u,v}$ to $b_{u,v}$ and an arc from $a_{u,v}$ to $b_{v,u}$,
 \item an arc from $u$ to $b_{u,v}$ and an arc from $v$ to $b_{v,u}$, and
 \item an arc from $b_{u,v}$ to $t_{col(u),col(v)}$ and an arc from $b_{v,u}$ to $t_{col(v),col(u)}$.
\end{itemize}
\begin{figure}[t]
\begin{center}
\includegraphics[width=\textwidth]{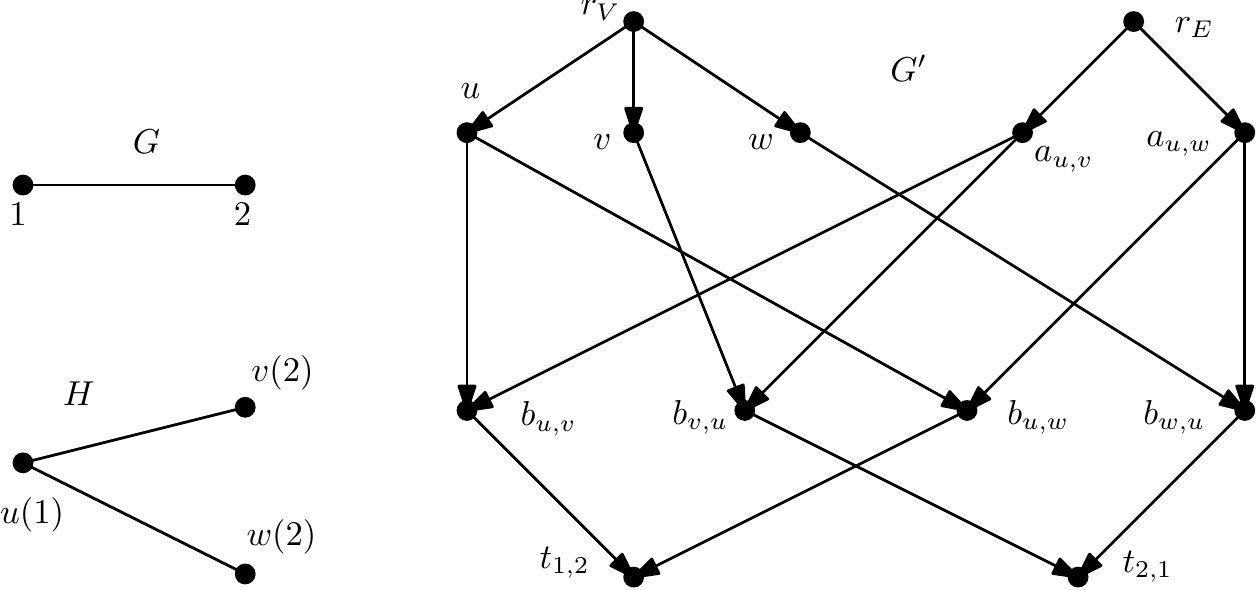}
\end{center}
 \caption{Illustration of the reduction in the proof of Theorem~\ref{thm:hardness}. The graphs $G$ and $H$ are on the left, $col$ is in the brackets for each vertex, and the graph $G'$ is on the right.}\label{fig:reduction}
\end{figure}

To finish the construction we let $k'=3|E_G|+|V_G|$. Note that we have $|T|=2|E_G|$, i.e., the new parameter depends linearly on the original one and the constructed graph is a directed acyclic graph.

We start the proof of correctness of the reduction by showing that if $(H,G,col)$ is a yes-instance of PSI, then $(G', R,T,k')$ is a yes-instance of 2-RST. Let $\phi:V_G\rightarrow V_H$ be the injection witnessing that $(H,G,col)$ is a yes-instance. We let $S=\{\phi(i) \mid i \in V_G\} \cup \{a_{\phi(i),\phi(j)},b_{\phi(i),\phi(j)},b_{\phi(j),\phi(i)} \mid \{i,j\} \in E_G, i <j\}$. Then for every $\{i,j\} \in E_G$, where $i <j$, we have $col(\phi(i))=i$, $col(\phi(j))=j$, $\{\phi(i),\phi(j)\} \in E_H$, and $\phi(i) < \phi(j)$. Therefore the vertices $r_V,\phi(i), b_{\phi(i),\phi(j)}, t_{i,j}$ form a path from $r_V$ to $t_{i,j}$ in $G'[R \cup S \cup T]$, the vertices $r_V,\phi(j), b_{\phi(j),\phi(i)}, t_{j,i}$ form a path from $r_V$ to $t_{j,i}$, the vertices $r_E,a_{\phi(i),\phi(j)}, b_{\phi(i),\phi(j)}, t_{i,j}$ form a path from $r_E$ to $t_{i,j}$, and the vertices $r_E,a_{\phi(i),\phi(j)}, b_{\phi(j),\phi(i)}, t_{j,i}$ form a path from $r_E$ to $t_{j,i}$ (here we use that $col(\phi(i))=i$, $col(\phi(j))=j$, and $\{\phi(i),\phi(j)\} \in E_H$). Since $|S|=|V_G|+3|E_G|=k'$, $(G', R,T,k')$ is indeed a yes-instance of 2-RST.

To prove the other implication, let us assume that $(G',R,T,k')$ is a yes-instance of $2$-RST and $S$ is the set witnessing it. As the vertices in $F$ are the only in-neighbors of vertices in $T$ and each vertex in $F$ has out-degree one, there must be at least $2|E_G|$ vertices of $F$ in $S$. There is a directed path from $a_{u,v}$ to $t_{i,j}$ in $G'$ if and only if $\{i,j\}=\{col(u),col(v)\}$. Hence there must be a separate vertex of $E'$ in $S$ for each edge of $G$. Similarly, there is a directed path from $u$ to $t_{i,j}$ in $G'$ if and only if $i=col(u)$. Hence, there is a separate vertex of $V_H$ in $S$ for each vertex of $G$ (note that $G$ is connected and has at least one edge). Therefore, as $|S| \leq 3|E_G| +|V_G|$, the budget is tight and there are exactly that many vertices from each of the sets.

For every $i \in V_G$ let $\phi(i)$ be the unique vertex $u$ of $V_H \cap S$ with $col(u)=i$. We show that $\phi$ has the desired properties. First of all it is injective and $col(\phi(i))=i$ for every $i \in V_G$ by definition. Next, note that a vertex $u \in S \cap V_H$ is connected to a vertex $b_{u',v}$ only if $u'=u$. Furthermore, for each $t_{i,j}$ there is a unique in-neighbor $b_{u,v}$ in $S$. Since there is a path from $r_V$ to each $t_{i,j}$, it follows that for every $b_{u,v} \in S$ we have $u \in S \cap V_H$, namely $u=\phi(col(u))$. 

Now each vertex in $E'$ has out-degree at most 2 and each vertex of $F$ has out-degree 1 in $G'[R\cup T \cup S]$. Thus, in order for $r_E$ to reach all $2|E_G|$ vertices in $T$, each of the $|E_G|$ vertices in $E' \cap S$ must have out-degree 2 in $G'[R\cup T \cup S]$, i.e., $b_{u,v}$ is in $S$ if and only if $b_{v,u}$ is. Hence, for every $b_{u,v} \in S$ we have also $v \in S \cap V_H$.
In other words, for every $\{i,j\} \in E_G$ there is an edge $\{\phi(i),\phi(j)\} \in E_H$, finishing the correctness of the reduction.

Since the new parameter $|T|$ is linear in the original $|E_G|$, the theorem now follows from the result of Pietrzak~\cite{Pietrzak03} and from \autoref{lem:psi hardness}.

The same holds also with respect to the parameter $(|T|+k)$.

\begin{corollary}
 $q$-RST is W[1]-hard with respect to $|T|$ even on directed acyclic graphs for every $q \ge 2$.
 Moreover, there is no algorithm for $q$-RST on directed acyclic graphs running in time $f(|T|)n^{o\left(\frac{|T|}{\log |T|}\right)}$ for any constant $q \ge 2$, unless ETH fails.
\end{corollary}

\begin{proof}
 It is enough to add $q-2$ vertices to $R$, each having an arc only to $r_V$. It is easy to verify that this preserves the properties of the construction. 
\end{proof}

\section{Restriction to Solutions with Pedestal}\label{sec:pedestal}
Having shown in the previous section that $q$-RST is still too general to allow for the nice algorithms known for DST and SCSS, in this section we restrict ourselves further.
To this end, we modify the definition of our problem in the sense that we do not require to obtain a path from each vertex of $R$ only to each vertex of $T$, but also to each other vertex of $R$.\\
\deftproblem{\textsc{$q$-Root Steiner Tree with Pedestal ($q$-RST-P)}}
{A directed graph $G=(V,A)$, two subsets of its vertices $R,T \subseteq V$ with $|R| = q$.}
{Find a minimum size of a set $S \subseteq V$ such that in $G[R \cup S \cup T]$ there is a directed path from $r$ to $t$ for every $r \in R$ and every $t \in R \cup T$.}

Note that this variant of $q$-RST could be also modeled by requiring $R \subseteq T$. However, to simplify the description, we assume $R \cap T = \emptyset$.

\begin{theorem}
For every $q \ge 1$ the problem $q$-RST-P is fixed-parameter tractable with respect to $|T|$. Namely there is an algorithm solving it in $O(2^{2q+4|T|}\cdot n^{2q+O(1)})$ time, where the constants hidden in the $O()$ notations are independent of $|T|$ and $q$.
\end{theorem}

The rest of this section is devoted to the proof of this theorem.

The $q$-RST-P problem with the set $T$ empty is exactly the SCSS problem (with $q$ terminals).
This problem was shown to be polynomial time solvable for every constant $q$ by Feldman and Ruhl~\cite{FeldmanR06} using a modeling by a token game. The cost of an optimal strategy for that game equals cost of the smallest solution to the SCSS instance.

We first slightly modify this game to model the problem for arbitrary $T$ in Subsection~\ref{sec:cautious}. We show there that optimal strategies for this game have some interesting properties which we can further use. Then, in Subsection~\ref{sec:accelerated}, we introduce a new game with more powerful moves which allows us to make many moves of the original game at once. Finally, in Subsection~\ref{sec:algo}, we show that the optimal strategies for the new game can be computed in the claimed running time.

\subsection{Cautious Token Game}\label{sec:cautious}
In this subsection we show how to modify the original token of Feldman and Ruhl in order to model the $q$-RST-P problem. 
We fix a vertex $r_0 \in R$ and let $R'=R \setminus \{r_0\}$. For a solution $S$ the graph $G[R\cup S\cup T]$ will contain a path from $r_0$ to $t$ for each vertex $t$ in $R' \cup T$. These paths together form an out-tree rooted at $r_0$ which is called the \emph{backward tree}. Also there is a path from each of the vertices in $R'$ to $r_0$, and these together form an in-tree rooted at $r_0$, called the \emph{forward tree}. 

The game traces the two trees by having three types of tokens, where two of them behave similarly. First, we have an $F$-token at each of the vertices of $R'$ and this token moves forward along the arcs of graph $G$. Second, we have a $B$-token at each vertex of $R'$, moving backward against the direction of the arcs of $G$. The third type of tokens we use (in difference to Feldman and Ruhl) are $D$-tokens which are originally placed one on each of the vertices of $T$ and move similarly as $B$-tokens. 

The purpose of the tokens is to trace the forward and backward tree. Hence, whenever two tokens of the same type arrive at the same vertex we can merge them to one token. This is also the case for $B$-tokens and $D$-tokens, and in case a $B$-token merges with a $D$-token we let the merged token be a $B$-token. The purpose of introducing the $D$-tokens is to show that these are somewhat less important for the game than $B$-tokens and, hence, they can be treated in a different way in the new game we will introduce in the next subsection.

The state of the game can be described by tree subset of vertices $(F,B,D)$ representing the set of vertices occupied by $F$-tokens, $B$-tokens, and $D$-tokens, respectively. Note that $|F| \le q, |B| \le q, |D| \le |T|$ during the whole game. Hence we take $F,B \in \binom{V}{\le q}$ and $D \in \binom{V}{ \le |T|}$ (here and on $\binom{V}{\le q}$ is the set of subsets of $V$ of size at most $q$).

The allowed moves are the following:
\begin{enumerate}[(1)]
 \item \emph{Single moves for respective tokens:} For every arc $(u,v) \in A$ and all sets $F,B \in \binom{V}{\le q}$ and $D \in \binom{V}{ \le |T|}$ we introduce the following moves:
  \begin{enumerate}
    \item If $u \in F$, then we have a move 
    $(F,B,D) \xrightarrow{c} ((F \setminus \{u\}) \cup \{v\},B,D),$ 
    where the cost $c$ of the move is 1 if $v \notin F\cup B \cup D$ and 0 otherwise.
    \item If $v \in B$, then we have a move
    $(F,B,D) \xrightarrow{c} (F,(B \setminus \{v\}) \cup \{u\},D \setminus \{u\}),$ 
    where the cost $c$ of the move is 1 if $u \notin F\cup B \cup D$ and 0 otherwise.
    \item If $v \in D$, then we have a move
    $(F,B,D) \xrightarrow{c} (F,B,(D \setminus \{v\}) \cup (\{u\} \setminus B)),$ 
    where the cost $c$ of the move is 1 if $u \notin F\cup B \cup D$ and 0 otherwise.
 \end{enumerate}
 \item \emph{Flipping:} For all sets $F,B \in \binom{V}{\le q}$ and $D \in \binom{V}{ \le |T|}$ we introduce the following moves:
 \begin{enumerate}
   \item if $F' \subseteq F$, $B' \subseteq B$, $D' \subseteq D$, $f \in F'$, and $b \in B'$,  
   then we have a move
   $(F,B,D) \xrightarrow{c} ((F\setminus F') \cup \{b\},(B\setminus B') \cup \{f\},D \setminus (D' \cup \{f\})),$ 
   where $c$ is the number of vertices on a shortest walk from $f$ to $b$ going through all vertices in $F' \cup B' \cup D'$. Here each vertex is counted each time it is visited, but vertices in $F' \cup B' \cup D'$ are counted once less.
   \item if $F' \subseteq F$, $B' \subseteq B$, $B' \neq \emptyset$, $D' \subseteq D$, $f \in F'$, and $d \in D'$, 
   then we have a move
   $(F,B,D) \xrightarrow{c} ((F\setminus F') \cup \{d\},(B\setminus B') \cup \{f\},D \setminus (D' \cup \{f\})),$ 
   where $c$ is the number of vertices on a shortest walk from $f$ to $d$ going through all vertices in $F' \cup B' \cup D'$. Here, again, each vertex is counted each time it is visited, but vertices in $F' \cup B' \cup D'$ are counted once less.
   \item if $F' \subseteq F$, $D' \subseteq D$, $f \in F'$, and $d \in D'$,  
   then we have a move
   $(F,B,D) \xrightarrow{c} ((F\setminus F') \cup \{d\},B,(D \setminus D')\cup (\{f\} \setminus B)),$ 
   where $c$ is the number of vertices on a shortest walk from $f$ to $d$ going through all vertices in $F' \cup D'$. as in the previous cases, each vertex is counted each time it is visited, but vertices in $F' \cup D'$ are counted once less.
  \end{enumerate}
\end{enumerate}

The original game of Feldman and Ruhl has only three types of moves: Single moves for $F$-tokens (exactly as (1-a)), single moves for $B$-tokens (similar as (1-b) and (1-c)) and flipping (all of (2)). If we did not distinguish the $B$-tokens and $D$-tokens (and consider all of them as $B$-tokens, we would get exactly this three types of moves. We make use of this fact in the proof of the equivalence of costs of optimal strategies for this game and sizes of solutions for the $q$-RST-P instance.

We distinguish the $B$- and $D$-tokens since we aim to show, e.g., that there is an optimal strategy for the game not using any moves of type (2-c). During the whole game the moves ensure that the invariant $D \cap B = \emptyset$ is maintained. This could be easily achieved by taking $D= D \setminus B$ after each move, however, we prefer to be more specific in taking out only the vertices which could actually newly appear in the intersection.

Now we would like to claim, that the game represents the instance $(G,R,T)$ of $q$-RST-P. Namely, that the minimum size of a solution to $(G,R,T)$ is exactly one less than the minimum cost of moves to get from $(R',R',T)$ to $(\{r_0\},\{r_0\},\emptyset)$ in the cautious token game. The easier direction is summarized by the following lemma (see also Lemma 3.1 of \cite{FeldmanR06}):

\begin{lemma}
 If there is a move sequence from $(R',R',T)$ to $(\{r_0\},\{r_0\},\emptyset)$ of total cost $c$, then there is a set $S \subseteq V$ of size at most $c-1$ such that in $G[R \cup S \cup T]$ there is a directed path from $r$ to $t$ for every $r \in R$ and every $t \in R \cup T$. Moreover, given the sequence, the corresponding set $S$ is easy to find.
\end{lemma}

The proof of this lemma follows from the definition of the moves of the game. If we let $S$ be the set of newly encountered vertices in the moves of the sequence excluding $r_0$ we get $|S|+1 \le c$, as the cost of each move is an upper bound on the number of newly encountered vertices including $r_0$. 

The next lemma provides the counterpart. Let us call a move of type (2) \emph{path-driven}, if the minimum size walk in the definition of the cost of the move can be taken as a simple path.

\begin{lemma}\label{lem:cautious}
 If there is a set $S \subseteq V$ of size at most $c-1$ such that in $G[R \cup S \cup T]$ there is a directed path from $r$ to $t$ for every $r \in R$ and every $t \in R \cup T$, then there is a move sequence from $(R',R',T)$ to $(\{r_0\},\{r_0\},\emptyset)$ of total cost at most~$c$ in which all type (2) moves are path-driven. 
\end{lemma}

The lemma can be proved using exactly the same arguments as for the proof of Lemma 2.2 in~\cite{FeldmanR06} (ignoring the difference between $B$- and $D$-tokens).
The aim is to construct a move sequence, where all intermediate position of tokens are in $H=G[R\cup S\cup T]$. 

Since each vertex is counted each time it becomes newly occupied by a token and we want to find a move sequence of cost $c$, we cannot  afford to re-occupy a previously abandoned vertex. Specifically, we enforce the following rule:
 \begin{align}
    \text{\parbox[t]{10cm}{Once a token moves off a vertex, no other token will ever move to that vertex again. } \tag{$\star$}}\label{prop:star}
 \end{align}
\newcommand{\prostar}{(\ref{prop:star})}
A vertex is called ``dead'' once a vertex moves from it and the tokens can only move to vertices which are still ``alive''. Note specifically that a token may be standing on a vertex that is already dead.

The sequence is constructed in a greedy fashion, maintaining the following invariant: 
\begin{align}
    \text{\parbox[t]{10cm}{There are paths using only alive vertices (except for endpoints) from each vertex of $F$ to $r_0$ and from $r_0$ to each vertex of $B \cup D$.} \tag{$\ast$}}\label{prop:ast}
\end{align}
\newcommand{\proast}{(\ref{prop:ast})}
This is actually the only assumption used in the correctness proof of Feldman and Ruhl to show that there is always a move to continue with, maintaining the properties \prostar{} and \proast{} (see the proof of Lemma 3.2 and Section 4 of \cite{FeldmanR06}). Since there is always a move to continue and no token can return to a vertex it has already visited, we must reach $(\{r_0\},\{r_0\},\emptyset)$ at some point and the cost cannot exceed $c$.

Since we aim on more detailed analysis of the game graph, we repeat some notions and lemmata used in the correctness proof given by Feldman and Ruhl. 

We say that a token $t$ \emph{requires} a vertex $v$ if all legal paths for $t$ to get to $r_0$ go through $v$. Here a legal path is a path in the right direction given by token $t$, within $H$, and using only alive vertices. We will further refer to a token and to the vertex it is currently standing on exchangeably, e.g., we say that a token $t$ requires a token $t'$ if $t$ requires the vertex $t'$ is currently standing on. Note that we can move the tokens using type (1) moves as long as they are not required by any other token. If each token is required by some other token, then the following holds.

\begin{lemma}[Flip Lemma, Lemma 3.3 of \cite{FeldmanR06}]
Let the $F_0$-tokens be the $F$-tokens that are not required by any other $F$-token. Similarly, let $B_0$-tokens be the $B$- and $D$-tokens, that are not required by any other $B$- or $D$-token. Suppose that every token is required by some other token. Then there is an $F_0$-token $f$ and a $B_0$-token $b$ such that 
\begin{itemize}
 \item $f$ requires $b$, and no other $F_0$-token requires $b$, and
 \item $b$ requires $f$, and no other $B_0$-token requires $f$. 
\end{itemize}
\end{lemma}

Let $f$ and $b$ be as in the flip lemma and $P$ be a simple path between them in $H$ using only alive vertices. Then we have the following claim.

\begin{claimx}[Claim 3.4 of \cite{FeldmanR06}]\label{cla:on_path}
All tokens that require a vertex on $P$ are on $P$ themselves.
\end{claimx}

If $F'$, $B'$, and $D'$ are the sets of $F$-, $B$-, and $D$-tokens on $P$, respectively, then we can apply type (2) move, preserving the property \proast. 

Since this is the only place where type (2) moves are used in the construction of the move sequence, we may assume that whenever a move of type (2) is used, it is used on the specific vertices $f$ and $b$ as selected by the flip lemma. 
Note that in this case the shortest walk from $f$ to $b$ going through all vertices in $F' \cup B' \cup D'$ is actually a simple path, i.e., all type (2) moves are path-driven. However, since it is complicated to test for the existence of such a path, following Feldman and Ruhl, we introduced more moves, which does not hurt the construction.

Our aim now is to show that the moves of type (2-c) can be omitted without affecting the correspondence between the game and the instance of $q$-RST-P.
\begin{lemma}
\label{lem:bez_2c}
If there is a set $S \subseteq V$ of size at most $c-1$ such that in $G[R \cup S \cup T]$ there is a directed path from $r$ to $t$ for every $r \in R$ and every $t \in R \cup T$, then there is a move sequence from $(R',R',T)$ to $(\{r_0\},\{r_0\},\emptyset)$ of total cost at most $c$ in which all type (2) moves are path-driven and, moreover, there are no moves of type (2-c). 
\end{lemma}

\begin{proof}
 Let $H= G[R \cup S \cup T]$ and let $H^*$ be an edge minimal subgraph of $H$ such that there is a directed path from $r$ to $t$ for every $r \in R$ and every $t \in R \cup T$. Let us construct the sequence greedily as in the proof of Lemma~\ref{lem:cautious} and suppose for contradiction that at some point a type (2-c) is to be applied on $f$, $d$, $F'$, and $D'$. That is, there is a path $P$ from $f$ to $d$ in $H^*$, $F'$ and $D'$ are the sets of $F$- and $D$-tokens on $P$, respectively, and there are no $B$-tokens on $P$.
 
 Now suppose that one of the tokens merged into the token $f$ started its tour on a vertex $s$ of $R'$, i.e., there is a path in $H^*$ from $s$ to $f$ using only dead vertices (except for $f$) (see also Fig.~\ref{fig:2c}). Let $b_s$ be the $B$-token that also started its tour on $s$. Since $b_s$ is not on $P$ and, hence, does not require $f$ by \autoref{cla:on_path}, there is a path from $r_0$ to $s$ in $H^*$ that avoids $f$. Note specifically that $f \neq s$, as this would mean that $s$ is required by $d$, $b_s$ stayed on $s$ and hence on $P$. Since $f$ requires $d$, there is a path of alive vertices from $d$ to $r_0$. Similarly, as $d$ requires $f$, there is a path of alive vertices from $r_0$ to $f$. Let $a$ be the last arc on that path. Note that it connects two alive vertices. 
 
 \begin{figure}[t]
 \begin{center}
 \includegraphics[width=.6\textwidth]{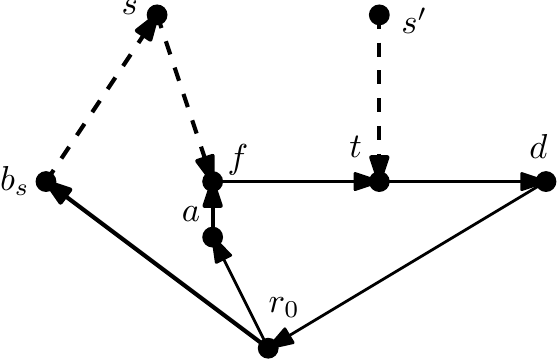}
 \end{center}
  \caption{Illustration of the situation in the proof of Lemma~\ref{lem:bez_2c}. The arrows denote paths and those dashed are formed by dead vertices.}\label{fig:2c}
 \end{figure}

 We claim that in $\hat{H} = H^* \setminus \{a\}$ there is also a directed path from $r$ to $t$ for every $r \in R$ and every $t \in R \cup T$, contradicting the minimality of $H^*$. In fact we will only show that there is a path in $\hat{H}$ from $r_0$ to each $t \in R' \cup T$ and from each $r$ in $R'$ to $r_0$. Such path are definitely present in $H^*$. If they avoid $f$, then they are still present in $\hat{H}$. If some of the paths cannot avoid $f$ in $\hat{H}$, then the corresponding token $t$ in the game for $H^*$ starting at vertex $s'$ requires $f$. It follows that $t$ lies on $P$ and there is a path between $s'$ and $t$ in the direction appropriate for the token formed by dead vertices (except for $t$). We know that $t$ is not a $B$-token, as there are no such tokens on $P$. If it is an $F$-token, then we can obtain a path from $t$ to $r_0$ avoiding $a$ by concatenating the part of $P$ to $d$ and a path from $d$ to $r_0$. If $t$ is a $D$-token, then we can obtain a path from $r_0$ to $t$ in $H^*$ by shortcutting the walk obtained by concatenating the path from $r_0$ to $s$ (avoiding $f$), the path from $s$ to $f$ (dead), and a part of $P$ (not containing $a$).
 
 Hence, indeed, if $H^*$ is minimal, then every type (2) move is of type (2-a) or (2-b).
\end{proof}

\newcommand{\fork}{{\cal Y}}
\autoref{lem:bez_2c} shows, that there are no flips, that would result in a $D$-token on a new position. We aim to show that $D$-tokens interact with the other tokens even less. 
Namely, if a $D$-token meets with an $F$-token, then it stays on place until it is merged with some $B$-token.
We show that by making a side step and considering the game on a modified graph.
To this end we need a following definition (see also Fig.~\ref{fig:forking}).

\begin{definition}\label{def:forking}
 Let $G =(V,A)$ be a directed graph. A \emph{forking} of $G$ is the directed graph $\fork(G)=(\hat{V},\hat{A})$, where 
 \begin{align*}
  \hat{V}&=V \times \{0,1\} \text{ and}\\
  \hat{A}&=\{((u,0),(v,1))\mid (u,v) \in A\} \cup \{((v,1),(v,0))\mid v \in V\}.
 \end{align*}
\end{definition}

\begin{figure}[t]
\begin{center}
\includegraphics[width=.8\textwidth]{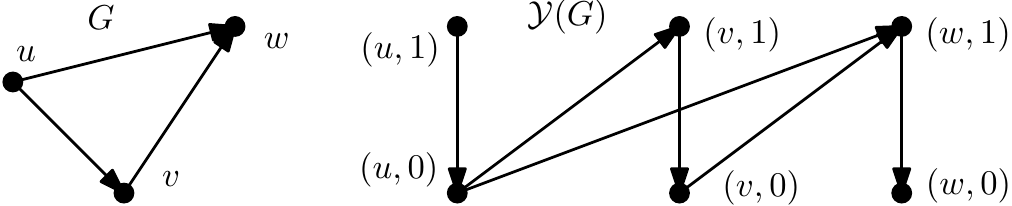}
\end{center}
 \caption{Illustration of the Definition~\ref{def:forking}---graph $G$ on left and $\fork(G)$ on right.}\label{fig:forking}
\end{figure}

There is a natural correspondence between paths in $G$ and paths in $\fork(G)$, namely a path $P$ in $G$ correspond to a path $\fork(P)$ in $\fork(G)$. Let us set $\hat{R}=\{(r,0)\mid r \in R\}$ and $\hat{T}=\{(t,0)\mid t \in T\}$. We relate the solution in the forking of $G$ to solution in $G$ by the following easy lemma.

\begin{lemma}
 Let $G=(V,A)$ be a directed graph, $R,T \subseteq V$, $R \cap T = \emptyset$, and $\fork(G)=(\hat{V},\hat{A})$ its forking. There is a solution $S \subseteq V \setminus (R \cup T)$ of size at most $k$ for $(G,R,T)$ if and only if there is a solution $\hat{S} \subseteq \hat{V}$ of size at most $2k + |R \cup T|$ for $(\fork(G),\hat{R},\hat{T})$.
\end{lemma}

\begin{proof}
 If $S$ is a solution for $(G,R,T)$ of size $k$, then let $\hat{S}= (S\cup R\cup T) \times \{0,1\} \setminus (\hat{R} \cup \hat{T})$. It is easy to verify, that $\hat{S}$ is a solution for $(\fork(G),\hat{R},\hat{T})$ of size at most $2k + |R \cup T|$.
 
 For the other direction, suppose $\hat{S}$ is a solution for $(\fork(G),\hat{R},\hat{T})$ of size at most $2k + |R \cup T|$.
 If there is $v \in V \setminus (R\cup T)$ such that $|\hat{S} \cap (\{v\} \times \{0,1\})|=1$, then the vertex in the intersection is either a sink or a source, as there are no arcs from $(v,1)$ and no arcs to $(v,0)$ except for $((v,1),(v,0))$. In both cases $\hat{S} \setminus \{(v,0),(v,1)\}$ is also a solution for $(\fork(G),\hat{R},\hat{T})$. Also note that since there is a path to every vertex in $\hat{R} \cup \hat{T}$ in $\fork(G)[\hat{R}\cup \hat{S} \cup \hat{T}]$, we have $(R \cup T) \times \{1\} \subseteq \hat{S}$. Hence we can assume that $\hat{R}\cup \hat{S} \cup \hat{T} = (R \cup S\cup T) \times \{0,1\}$ for some $S \subseteq V$ of size at most $k$. It is easy to verify that this $S$ forms a solution for $(G,R,T)$.
\end{proof}

Now consider the token game for the forking of $G$, $\hat{R}$, and  $\hat{T}$ and suppose that there is a solution $\hat{S}$ for $(\fork(G),\hat{R},\hat{T})$ of size at most $2k + |R \cup T|$. By \autoref{lem:bez_2c}, there is a sequence $M$ of moves from $(\hat{R} \setminus \{(r_0,0)\}, \hat{R} \setminus \{(r_0,0)\}, \hat{T})$ to $(\{(r_0,0)\},\{(r_0,0)\}, \emptyset)$ of total cost $2k + |R \cup T|+1$ without moves of type (2-c). 

We want to show that the following claim.
\begin{claimx}
 If there is a vertex $v \in V$ such that at some point in the move sequence a $D$-token and an $F$-token each occupy a vertex in $\{v\} \times \{0,1\}$, then the next move touching the $D$-token is either of type (2), or single-move (type (1)) of some $B$-token merging with the $D$-token.
\end{claimx}
\begin{proof}
First of all if the tokens actually meet on the same vertex, say $(v,i)$, then at least one of the tokens must have used the arc $((v,1),(v,0))$ and the vertex $(v,1-i)$ is dead. This, however, contradicts the property \proast{} for the other token. Now suppose there is an $F$-token on $(v,1)$ and a $D$-token on $(v,0)$. Since a type (2-a) or (2-b) moves do not place $D$-tokens on vertices where they were not previously, the previous move touching the $D$-token was of type (1). Therefore, before the move, the $D$-token was on $(v,1)$. If the $F$-token was not there before that move, then the property \prostar{} is violated, otherwise they meet at the vertex $(v,1)$ which is not possible as we have already shown. Finally, if there is an $F$-token on $(v,0)$ and a $D$-token on $(v,1)$, then the tokens require each other and, therefore, no single move can be applied to any of them. Thus the claim follows.
\end{proof}

We have shown that there is an optimal move sequence from $(\hat{R} \setminus \{(r_0,0)\}, \hat{R} \setminus \{(r_0,0)\}, \hat{T})$ to $(\{(r_0,0)\},\{(r_0,0)\}, \emptyset)$ with the property given by the claim. Now we want to show, that this move sequence can be translated to an optimal move sequence from $(R\setminus \{r_0\},R\setminus \{r_0\},T)$ to $(\{r_0\},\{r_0\},\emptyset)$ in the original graph, with the property given by the claim. 

We achieve that by simply projecting everything to the first coordinate. I.e., let us replace single moves for arcs $((u,0),(v,1))$ by single moves for the arc $(u,v)$ and omit single moves for arcs $((v,1),(v,0))$. If there is a flip move with the path being just the arc $((v,1),(v,0))$ for some $v$, then we just omit it. If the path is longer, we replace the flip by a flip on the projection of the path and with the projections of the sets $F'$, $B'$, $D'$. 

We should verify that this produces a valid move sequence. It is enough to show that the property \prostar{} is not violated, the property \proast{} then follows. Suppose that some token leaves the vertex $v$ before some other vertex enters it. But this would mean that in $\fork(G)$, a token leaves $(v,i)$ before some other token enters $(v,1-i)$. However, as each token has to visit both vertices $(v,i)$ and $(v, 1-i)$ this would violate the property \prostar{} also in the graph $\fork(G)$.

Finally, note that the produced sequence has still the property that it contains no (2-c) moves, and, whenever an $F$-token meets with a $D$-token on a vertex, then the next move touching the $D$-token is either of type (2), or single-move (type (1-b)) of some $B$-token merging with the $D$-token. This is not affected by omitting some of the flips as there are no moves of type (2-c).
In other words, the vertex must stay alive until the $D$-token is touched and, hence, the $F$-token will not  be touched before the $D$-token is touched, except possibly by type (1-a) moves of other $F$-tokens merging with the current one.

Move sequences with these properties allow us to postpone the moves of the $D$-tokens just before the move in which these tokens are merged with some $B$-token. This is a crucial property we use in the accelerated game. 

We formulate the obtained results as a lemma.

\begin{lemma}\label{lem:krasne_sekvence}
If there is a solution $S \subseteq V$ of size at most $c-1$ for $(G,R,T)$, then there is a move sequence from $(R',R',T)$ to $(\{r_0\},\{r_0\},\emptyset)$ of total cost at most $c$ in which all type (2) moves are path-driven and, moreover, there are no moves of type (2-c). 
Furthermore, in this sequence of moves, whenever after some move an $F$-token and a $D$-token sit together on a vertex $v \in V$, then the next move touching the $D$-token is either of type (2), or single-move (type (1-b)) of some $B$-token merging with the $D$-token.
\end{lemma}

\subsection{Accelerated Token Game}\label{sec:accelerated}
In the accelerated game the $D$-tokens stay at their places until the move in which they should be merged with a $B$-token. The moves then also include the costs of moving the $D$-tokens from their original places to the vertex where they get merged with the $B$-token. Therefore, we now represent the positions of the $D$-tokens only as subsets of $T$.

To define the costs of the moves we use the following notion. Let $ST(r,X)$ be the minimum number of vertices in a set $S$ such that in $G[\{r\} \cup X \cup S]$ there is a path from $r$ to every $x \in X$. I.e., this is a variation of DST with root $r$ and terminals $X$.

We have the following moves (we number the moves from (3), as not to confuse them with the moves of the cautious game).
\begin{enumerate}[(1)]
 \setcounter{enumi}{2}
 \item \emph{Single moves}: For every arc $(u,v) \in A$ and all sets $F,B \in \binom{V}{\le q}$ and $D \subseteq T$ we introduce the following moves:
  \begin{enumerate}
    \item If $u \in F$, then we have a move
    $(F,B,D) \xrightarrow{c} ((F \setminus \{u\}) \cup \{v\},B,D),$ 
    where the cost $c$ of the move is 1 if $v \notin F\cup B \cup D$ and 0 otherwise.
    \item If $v \in B$ and $D' \subseteq D$ then we have a move
    $(F,B,D) \xrightarrow{c} (F,(B \setminus \{v\}) \cup \{u\},D \setminus (D' \cup \{u\})),$ 
    where the cost $c$ is $ST(u,D' \cup \{v\})+1$ if $u \notin F \cup B \cup D$ and $ST(u,D' \cup \{v\})$ otherwise.      
 \end{enumerate}
 \item \emph{Flipping}: For all sets $F,B \in \binom{V}{\le q}$ and $D \subseteq T$ we introduce the following moves:
 \begin{enumerate}
   \item if $F' \subseteq F$, $B' \subseteq B$, $D' \subseteq D$, $f \in F'$, and $b \in B'$, then we have a move
   $(F,B,D) \xrightarrow{c} ((F\setminus F') \cup \{b\},(B\setminus B') \cup \{f\},D \setminus (D' \cup \{f\})),$ 
   where $c$ is as described below.
   
   \item if $F' \subseteq F$, $B' \subseteq B$, $B' \neq \emptyset$, $D' \subseteq D$, $D' \neq \emptyset$, $f \in F'$, and $v$ is an arbitrary vertex in $V \setminus B$,  then we have a move
   $(F,B,D) \xrightarrow{c} ((F\setminus F') \cup \{v\},(B\setminus B') \cup \{f\},D \setminus (D' \cup \{f\})),$ 
   where $c$ is as described below.
  \end{enumerate}
 \item \emph{Finishing}: For all sets $D \subseteq T$ we have a move
   $(\{r_0\},\{r_0\},D)\xrightarrow{c} (\{r_0\},\{r_0\},\emptyset),$
   where $c= ST(r_0,D)$.
 
\end{enumerate}
Let us now explain the intuition behind the moves.  It is clear for (3-a). We consider the move (3-b) to move the $B$-token from $v$ as well as the $D$-tokens from $D'$ to $u$. The move (4-a) moves all $B$- and $D$-tokens from $B' \cup D'$ to $f$ at the same time moving the $F$-tokens from $F'$ to $b$. The move (4-b) does the same thing, except that the $F$-tokens are taken to a vertex $v$. The move (5) moves the $D$-tokens from $D$ to $r_0$.

We would like to define the cost of moves of type (4) as the minimum number of vertices in a subgraph that provides a walk from $f$ to $b$ (or from $f$ to $v$) through all vertices in $F' \cup B'$ and at the same time a path from $f$ to each vertex of $D'$, where the vertices in $F' \cup B' \cup D'$ again do not count. In fact the solution will again use this type of moves only when there is a simple path from $f$ to $b$ (or from $f$ to $v$) through all vertices in $F' \cup B'$. 

As this condition is complicated to test and the desired cost is complicated to compute, we will define a cost of a move which provides an upper bound on the desired cost, and coincides with the desired cost whenever the optimal walk is actually a simple path. 

We define the cost of the type (4-a) moves $c$ to be the minimum over all bijections $\phi:\{2, \ldots, |F' \cup B'|-1\} \to (F' \cup B') \setminus \{f,b\}$ (representing the order of the vertices of $(F' \cup B') \setminus \{f,b\}$ along the walk) and all mappings $\psi: D' \to (F' \cup B')\setminus \{b\}$ (representing the part of the path at which the particular $D$-token joins it) of the sum $\sum_{i=1}^{|F' \cup B'|-1} ST(\phi(i),\psi^{-1}(\phi(i))\cup \{\phi(i+1)\}),$
where $\phi(1)=f$ and $\phi(|F' \cup B'|)=b$.

Similarly, the cost of a (4-b) move is the minimum over all bijections $\phi: \{2, \ldots, |F' \cup B'|\} \to (F' \cup B') \setminus \{f\}$ and all mappings $\psi: D' \to (F' \cup B')$ of the sum $\sum_{i=1}^{|F' \cup B'|} ST(\phi(i),\psi^{-1}(\phi(i))\cup \{\phi(i+1)\}),$
where $\phi(1)=f$ and $\phi(|F' \cup B'|)=v$. Here, the cost is increased by one if $v \notin (D \cup F)$.


We are going to show that solving the instance of $q$-RST-P again corresponds to finding a cheapest possible move sequence from $(R',R', T)$ to $(\{r_0\},\{r_0\},\emptyset)$ in the accelerated token game.

Again, it is easy to see that if there is a move sequence, then there is a solution of the corresponding cost.

\begin{lemma}\label{lem:e_ag_ms}
 If there is a move sequence of the accelerated game from $(R',R',T)$ to $(\{r_0\},\{r_0\},\emptyset)$ of total cost $c$, then there is a solution $S \subseteq V$ for $(G,R,T)$ of size at most $c-1$.  
\end{lemma}

As in the previous cases, given the move sequence, the solution can be obtained as a union of the vertex sets witnessing the costs of particular moves.

Let us now focus on the more complicated direction of the equivalence.

\begin{lemma}
\label{lem:e_ms_ag}
If there is a solution $S \subseteq V$ for $(G,R,T)$ of size at most $c-1$, then there is a move sequence of the accelerated game from $(R',R',T)$ to $(\{r_0\},\{r_0\},\emptyset)$ of total cost at most $c$.
\end{lemma}

\newcommand{\dvar}{D_{var}}

\begin{proof}
By \autoref{lem:krasne_sekvence} there is a move sequence $M$ for the cautious game from $(R',R',T)$ to $(\{r_0\},\{r_0\},\emptyset)$ of total cost at most $c$, which, moreover, does not use moves of type (2-c) and such that whenever after some move an $F$-token and a $D$-token sit together on a vertex $v \in V$, then the next move touching the $D$-token is either of type (2), or single-move (type (1-b)) of some $B$-token merging with the $D$-token.

We construct a move sequence $M'$ for the accelerated game based on the sequence $M$. We construct the sequence from the start to end and during the process we keep updated a variable $\dvar$ representing the remaining unmoved $D$-tokens. Initially it is set to $T$. Let us take the moves of the sequence $M$ one by one. While the moves of type (1-a), (1-b), (2-a), and (2-b) are immediately translated to moves of $M'$ with the same sets $F$ and $B$, the moves of type (1-c) are postponed and included in the next move of other type that touches the particular token. 

Since we want to show that the sequence $M'$ has also cost at most $c$, we charge some vertices to pay the cost of the move currently added to $M'$. These are always the vertices newly occupied and sometimes some more vertices which we turn dead. Each vertex in $S \cup \{r_0\}$ can pay only once during the game and only a cost of 1. We will show at the end that this condition is not violated.

If there is a (1-a) move $(F,B,D) \rightarrow (F',B,D)$ in $M$, then we put into $M'$ a (3-a) move $(F,B,\dvar) \rightarrow (F',B,\dvar)$. If the cost of the move is 1, then we charge the unique vertex in $F' \setminus F$.

If there is a (1-b) move $(F,B,D) \rightarrow (F,B',D')$ for arc $(u,v)$, then let $D_v \subseteq T$ be the set of original locations of the $D$-tokens merged (by (1-c) moves) into the $B$-token on $v$ since it was placed there. Also, if $u \in D$, then let $D_u \subseteq T$ be the set of original locations of tokens merged into the $D$-token on $v$. Otherwise let $D_u=\emptyset$. We put into $M'$ the (3-b) move $(F,B,\dvar) \rightarrow (F,B',\dvar \setminus (D_v \cup D_u))$ and update $\dvar$ to  $\dvar \setminus (D_v \cup D_u)$. 

We charge the vertices involved in the (1-c) moves moving the $D$-tokens from $D_v$ to $v$ and from $D_u$ to $u$, except for $D_u \cup D_v \cup \{v\}$. The vertex $u$ is only charged if $u \notin F \cup B \cup \dvar$. Note that this way we did not charge any vertex of $F \cup B \cup \dvar$ as, due to the properties of $M$, the $D$-tokens can only meet the $B$- or $F$-tokens in $v$ or $u$, and all $D$-tokens which would merge with the $D$-tokens considered are also considered. Since the charged vertices together with vertices in $F \cup B \cup \dvar$ provide paths from $u$ to all vertices in $(D_v \cup D_u \cup \{v\})$, it follows that the number of charged vertices is at least $ST(u, D_v \cup D_u \cup \{v\})+1$ if $u \notin F \cup B \cup \dvar$ and $ST(u,D_v \cup D_u \cup \{v\})$ otherwise.

If there is a (1-c) move in $M$, then we skip it and put no new moves into $M'$ at the moment.

For a (2-a) move $(F,B,D) \rightarrow ((F\setminus F') \cup \{b\},(B\setminus B') \cup \{f\},D \setminus (D' \cup \{f\}))$ in $M$ let us do the following. For each $v \in (D' \cup \{f\})$ let $D_v \subseteq T$ be the set of original locations of the $D$-tokens merged into the token on $v$. Similarly, for each $v \in B'$ let $D_v \subseteq T$ be the set of original locations of the $D$-tokens merged into the $B$-token on $v$ since it was placed there. We let $D''= \bigcup_{v \in D' \cup B' \cup \{f\}}D_v$ (note that the union is disjoint). We put into $M'$ the (4-a) move \[(F,B,\dvar) \rightarrow ((F\setminus F') \cup \{b\},(B\setminus B') \cup \{f\},\dvar \setminus D'').\] Then we update $\dvar$ to $\dvar \setminus D''$.

To charge some vertices for the cost of the move, let $P$ be the path driving the cost of the (2-a) move and $\ell=|F' \cup B'|$. Let $\phi:\{1, \ldots, \ell\} \to (F' \cup B')$ be the order of the vertices in $F' \cup B'$ along $P$ from $f$ to $b$. That is $\phi(1)=f$, $\phi(2)$ is the next vertex from $F' \cup B'$ along the path etc. For each $i \in \{1, \ldots, \ell-1\}$ let $P_i$ be the vertices on $P$ between $\phi(i)$ and $\phi(i+1)$ and $X'_i$ be union over $v \in (P_i \cup \{\phi(i)\}) \cap (D' \cup B' \cup \{f\})$ of the sets of vertices involved in the (1-c) moves in $M$ moving the $D$-tokens from $D_v$ to $v$. Let $X'_{\ell}$ be the set of vertices involved in the (1-c) moves in $M$ moving the $D$-tokens from $D_b$ to $b$. For each $i \in \{1, \ldots, \ell-2\}$ let $X_i = X'_i \cup P_i$ and let $X_{\ell-1} = X'_{\ell-1} \cup X'_{\ell} \cup P_{\ell-1}$. We charge all vertices in $\bigcup_{i=1}^{\ell-1}X_i$ except for vertices in $F' \cup B' \cup D''$. Note again that the union is disjoint.

We should show that the number of charged vertices is at least the cost of the move. Let us define a mapping $\psi: D'' \to (F' \cup B')\setminus \{b\}$ as follows. For each $i \in \{1, \ldots, \ell-1\}$, each $v \in (P_i \cup \{\phi(i)\}) \cap (D' \cup B')$, and each $d \in D_v$ let $\psi(d)=\phi(i)$ and for each $d \in D_b$ let $\psi(d)=\phi(\ell-1)$. Since for each $i \in \{1, \ldots, \ell-1\}$, the vertices in $X_i \cup \{\phi(i),\phi(i+1)\} \cup \psi^{-1}(\phi(i))$ provide a path from $\phi(i)$ to each vertex in $\psi^{-1}(\phi(i)) \cup \{\phi(i+1)\}$, we have $|X_i \setminus (\{\phi(i),\phi(i+1)\} \cup \psi^{-1}(\phi(i)))| \ge ST(\phi(i),\psi^{-1}(\phi(i))\cup \{\phi(i+1)\})$ and the number of charged vertices is at least the cost of the move.

For a (2-b) move we proceed similarly as for (2-a) moves with small differences. Let $(F,B,D) \xrightarrow{c} ((F\setminus F') \cup \{d\},(B\setminus B') \cup \{f\},D \setminus (D' \cup \{f\}))$ the (2-b) move in $M$. For each $v \in (D' \cup \{f\})$ let $D_v \subseteq T$ be the set of original locations of the $D$-tokens merged into the token on $v$. Similarly, for each $v \in B'$ let $D_v \subseteq T$ be the set of original locations of the $D$-tokens merged into the $B$-token on $v$ since it was placed there. We let $D''= \bigcup_{v \in D' \cup B' \cup \{f\}}D_v$ (the union is disjoint). We put into $M'$ the (4-b) move \[(F,B,\dvar) \rightarrow ((F\setminus F') \cup \{d\},(B\setminus B') \cup \{f\},\dvar \setminus D''),\] i.e., the move is performed with $d$ in the role of $v$. Then we update $\dvar$ to $\dvar \setminus D''$.

To charge some vertices for the cost of the move, let $P$ be the path driving the cost of the (2-b) move and $\ell=|F' \cup B' \cup \{d\}|$. Let $\phi:\{1, \ldots, \ell\} \to (F' \cup B' \cup \{d\})$ be the order of the vertices in $F' \cup B' \cup \{d\}$ along $P$ from $f$ to $d$. The sets $P_i$, $D_v$, and $X'_i$ are defined the same way as for the (2-a) moves and  $X'_{\ell}$ is the set of vertices involved in the (1-c) moves in $M$ moving the $D$-tokens from $D_d$ to $d$. We define the sets $X_i$ the same way as before and again charge all vertices in $\bigcup_{i=1}^{\ell-1}X_i$ except for vertices in $F' \cup B' \cup D'$. The vertex $d$ is charged if and only if it is not in $D \cup F$.

We should again show that the number of charged vertices is at least the cost of the move. Let us define the mapping $\psi: D'' \to (F' \cup B')\setminus \{d\}$ the same way as above. Specifically, for each $d' \in D_d$ let $\psi(d')=\phi(\ell-1)$. Since for each $i \in \{1, \ldots, \ell-1\}$, the vertices in $X_i \cup \{\phi(i), \phi(i+1)\} \cup \psi^{-1}(\phi(i))$ again provide a path from $\phi(i)$ to each vertex in $\psi^{-1}(\phi(i)) \cup \{\phi(i+1)\}$, we have $|X_i \setminus (\{\phi(i),\phi(i+1)\} \cup \psi^{-1}(\phi(i)))| \ge ST(\phi(i),\psi^{-1}(\phi(i))\cup \{\phi(i+1)\})$ and the number of charged vertices is at least the cost of the move.

Finally, if the whole sequence was processed and $\dvar \neq \emptyset$, then we add at the end of the sequence $M'$ a type (5) move $(\{r_0\},\{r_0\},\dvar)\rightarrow (\{r_0\},\{r_0\},\emptyset)$ and charge the vertices involved in the (1-c) moves in $M$ moving  the $D$-tokens from $\dvar$ to $r_0$ except from vertices in $\dvar$ and $r_0$.

It remains to show that each vertex is charged at most once. For that purpose we distinguish the vertices only used by $D$-tokens and the vertices also used by an $F$- or $B$-token. Observe that the tokens use exactly the same vertices in the sequence $M'$ as in $M$. If two $D$-tokens used the same vertex not used by $F$- or $B$-tokens, then these were merged together in $M$ and hence processed in the same move in $M'$, therefore the vertex was only charged once. 

Among the vertices used by an $F$- or $B$-token the move of $M'$ always charges the vertices newly occupied by such tokens and those turned dead by the corresponding move of $M$ without being previously occupied by an $F$- or $B$-token. Since each vertex gets newly occupied by an $F$- or $B$-token at most once and only if it is not dead and can be turned dead only once in $M$, it follows that each vertex is charged at most once. 
\end{proof}

\subsection{The Algorithm}\label{sec:algo}

The algorithm first computes the so-called ``game graph'' for the accelerated game. The vertex set $V'$ of this directed graph is formed by all possible configurations of the tokens in the game, i.e., by all triples $(F,B,D)$, where $F,B \in \binom{V}{\le q}$ and $D \subseteq T$. The arcs $A'$ of this graph correspond to legal moves of the accelerated game, i.e., moves of type (3),(4), and (5). The length of each arc is equal to the cost of the corresponding move of the game.

Once the game graph is constructed, we simply find the shortest path from the vertex $(R',R', T)$ to the vertex $(\{r_0\},\{r_0\},\emptyset)$. Since each arc of the graph corresponds to a move of the game, the shortest path corresponds to an optimal sequence of moves to get from the configuration $(R',R', T)$ to the configuration $(\{r_0\},\{r_0\},\emptyset)$. By Lemmata~\ref{lem:e_ag_ms} and~\ref{lem:e_ms_ag} such a sequence of cost $c$ exists if and only if there is a solution of size $c-1$ for the instance $(G,R,T)$ of $q$-RST-P.

To bound the running time of the algorithm, let us first focus on computing the costs of the moves. For that purpose we often need to compute $ST(r,X)$. To this end we can use the algorithm of Nederlof~\cite{Nederlof13} running in $O(2^{|X|}\cdot n^{O(1)})$ time (see also Misra et al.~\cite[Lemma 2]{MisraPRSS12}).

The game graph has at most $\binom{n}{\le q}^2 \cdot 2^{|T|}=O(n^{2q}\cdot 2^{|T|})$ vertices. There are at most $O(n^{2q-1}\cdot m\cdot 2^{|T|})$ moves of type (3-a) and at most $n^{2q-1}\cdot m\cdot 3^{|T|}$ moves of type (3-b) in total (each vertex of $T$ can be either in $T\setminus D$, in $D \setminus D'$, or in~$D'$). Computing the cost of a (3-a) move is trivial and can be done in $O(q+|T|)$ time. For (3-a) move we need to run the algorithm of Nederlof once, which takes $O(2^{|T|}\cdot n^{O(1)})$ time.

Concerning type (4) moves, for each vertex $(F,B,D)$ there are at most $2^{|F|}\cdot |F|\cdot 2^{|B|}\cdot |B|\cdot 2^{|D|}$ moves of type (4-a) and at most $2^{|F|}\cdot |F|\cdot 2^{|B|}\cdot (n-|B|)\cdot 2^{|D|}$ moves of type (4-b). This gives at most $O(n^{2q+1}\cdot 2^{2|T|}2^{2q})$ of them in total. 

To compute the cost of a type (4) move, we may use dynamic programming of Held-Karp type~\cite{HK62}. Let us describe it in more detail for type (4-a) move. We introduce one table indexed by subsets $J$ of $F' \cup B'$ which are proper supersets of  $\{f\}$, a vertex $j$ from $J\setminus \{f\}$ and a subset $D''$ of $D'$. The semantic meaning is that we compute the cost as if only the vertices in $J$ and $D''$ were involved in the move and the $F$-tokens from $J$ ended in vertex $j$.

The table is initialized by setting the entries indexed by $(\{f,j\},j,D'')$ to $ST(f,D'' \cup \{j\})$.
Then we fill the table according to the size of $J$ starting from the smallest sets.
For every $J$, with $|J| > 2$, $j \in J\setminus \{f\}$, and $D''\subseteq D'$ we set the entry indexed by $(J,j,D'')$ to the minimum over all $j' \in J\setminus \{j,f\}$ and all subsets $\bar{D}$ of $D''$ of the sum of $ST(j',\bar{D} \cup \{j\})$ and the value on entry indexed $(J\setminus \{j\}, j', D'' \setminus \bar{D})$. The resulting cost of the (4-a) move is then found in the entry indexed by $(F' \cup B', b, D')$.

It is straightforward to show that the above described dynamic programming computes the cost of the move as desired. The only difference for type (4-b) moves is that $J$ is a subset of $F' \cup B' \cup \{v\}$. The table has at most $2^{|F' \cup B' \cup \{v\}|-1+|D'|}$ entries, we check at most $(|F' \cup B' \cup \{v\}|-1) \cdot 2^{D'}$ combinations of $j'$ and $\bar{D}$, and for each of them call the Nederlof's algorithm in $O(2^{|\bar{D}|}\cdot n^{O(1)})$ time. This gives $O(2^{2q+3|T|}\cdot n^{O(1)})$ time to compute the cost of the move. In fact, taking the table for $F'=F$, $B'=B$, and $D'=D$ one can determine the cost of each possible (4-a) move from the configuration $(F,B,D)$ by simply looking to the appropriate entry of the table. This can be done similarly for (4-b) moves and each $v$. Therefore, computing the costs of all type (4) moves takes $O(2^{2q+4|T|}\cdot n^{2q+O(1)})$ time in total.

Finally, there are at most $2^{|T|}$ moves of type (5), and computing the cost of each of them amounts to single run of the Nederlof's algorithm, thus in total taking $O(2^{2|T|}\cdot n^{O(1)})$ time.

The game graph has in total $O(n^{2q-1}\cdot m\cdot 2^{|T|} + n^{2q-1}\cdot m\cdot 3^{|T|} + n^{2q+1}\cdot 2^{2|T|}2^{2q}+2^{|T|})=O(n^{2q+1}\cdot 2^{2|T|+2q})$ arcs. Hence, the Dijkstra's algorithm~\cite{Dijkstra59} using the Fibonacci heaps~\cite{FredmanT87}, with running time $O(|V'| \log |V'|+|A'|)$ runs in $O(n^{2q}\cdot 2^{|T|} \cdot \log (n^{2q}\cdot 2^{|T|})+ n^{2q+1}\cdot 2^{2|T|+2q})=O(n^{2q+1}\cdot 2^{2|T|+2q})$ time. 

To sum up, the running time of the algorithm is determined by the computation of the costs of type (4) moves, which takes $O(2^{2q+4|T|}\cdot n^{2q+O(1)})$ time. 

In fact the constants in the above running time can be slightly reduced by storing some values of $ST()$ and by a more detailed analysis which we ommit here. 

\section{Restriction to Solutions with Trunk}\label{sec:trunk}
In this section we relax the conditions on the solutions of the problem, namely, we consider the following problem:\\
\deftproblem{\textsc{$q$-Root Steiner Tree with Trunk ($q$-RST-T)}}
{A directed graph $G=(V,A)$, two subsets of its vertices $R,T \subseteq V$ with $|R| = q$.}
{Find a minimum size of a set $S \subseteq V$ such that in $G[R \cup S \cup T]$ there is a vertex $v$, a directed path from $r$ to $v$ for every $r \in R$, and a directed path from $v$ to $t$ for every $t \in T$.}
We show that this problem can be solved in similar time as $q$-RST-P.

\begin{theorem}
\label{thm:trunk}
For every $q \ge 1$ the problem $q$-RST-T is fixed-parameter tractable with respect to $T$. Namely, there is an algorithm solving it in $O(2^{2q+4|T|}\cdot n^{3q+O(1)})$ time. 
\end{theorem}

\begin{proof}
We reduce the problem to the previous one in the following way. Suppose there was a solution $S$ to the problem and $v$ was the vertex with the specified properties. For each vertex $r$ in $R$ we can identify the first vertex $r'$ on a path from $r$ to $v$ in $G[R \cup S \cup T]$ that is reachable from $v$ in the same graph. Let us denote $R'$  the set of all such $r'$'s. Then, the solution can be divided into parts providing paths from $R$ to $R'$ and to a $q$-RST-P instance $(G, R', T \setminus R')$. 

To apply the above idea, the algorithm first guesses a mapping $\phi: R \to V$, i.e., tries all possible such mappings. For each of them let us denote $R'= \phi(R)$. Now for each $\phi$ we compute the minimum size of a solution for the $|R'|$-RST-P instance $(G,R',T \setminus R')$. We increase this value by $|R'\setminus (R \cup T)|$ and by the sum $\sum_{r' \in R'} STR(r', \phi^{-1}(r'))$, where $STR(x, Y)$ is a minimum size of a vertex set $S$ such that in $G[\{x\} \cup S\cup Y]$ there is a path from each vertex of $Y$ to $x$, i.e., a variant of directed Steiner in-tree. We output the minimum of the values obtained this way over all $\phi$'s.

To see that this algorithm correctly computes the minimum size of a set $S \subseteq V$ such that in $G[R \cup S \cup T]$ there is a vertex $v$, a directed path from $r$ to $v$ for every $r \in R$, and a directed path from $v$ to $t$ for every $t \in T$, let us first show that the value computed by the algorithm is at least the optimal value for the instance. To this end, it is enough to observe that for each $\phi$ the Steiner in-trees computed in the sum together with the set $R'$ provide a path from $r$ to $\phi(r)$ for each $r \in R$. Moreover, the solution to the instance $|R'|$-RST-P instance $(G,R',T \setminus R')$ provides paths from $r'$ to $t$ for each $r' \in R'$ and each $t \in T \cup R'$. Therefore, denoting $v= r'$ for any $r' \in R'$, these sets together provide a solution for the $q$-RST-T instance and the optimal value for the instance must be at most the one computed by the algorithm.

For the other direction, let us now repeat more formally the idea from the beginning of the proof. Suppose $S \subseteq V$ is a set of minimum size such that in $G[R \cup S \cup T]$ there is a vertex $v$, a directed path from $r$ to $v$ for every $r \in R$, and a directed path from $v$ to $t$ for every $t \in T$. Let $H$ be the subgraph of $G[R \cup S \cup T]$ with the minimum number of arcs such that in $H$ there is still a directed path from $r$ to $v$ for every $r \in R$, and a directed path from $v$ to $t$ for every $t \in T$. Now for each $r \in R$ let us denote $\phi(r)$ the first vertex on some path from $r$ to $v$ in $H$ reachable from $v$ in $H$ and denote $R'= \phi(R)$. 

For each $r' \in R'$ denote $X_{r'}$ the set of vertices on the paths in $H$ from the vertices in $\phi^{-1}(r')$ to  $r'$, except for the vertices in $\phi^{-1}(r')$ and $r'$ itself and $X = \bigcup_{r' \in R'} X_{r'}$. We let $Y = S \setminus (R' \cup X)$. We would like to show that the sets $X_{r'}\cup \phi^{-1}(r') \cup \{r'\}$ are disjoint for distinct $r'$'s and that $H[R' \cup Y \cup T]$ contains a path from each vertex $r' \in R'$ to each vertex $t \in R' \cup T$.  

If there were $x,y \in R'$, $x \neq y$ and a vertex $a \in (X_{x}\cup \phi^{-1}(x)) \cap (X_{y}\cup \phi^{-1}(y))$, then there are two paths from $a$ to $v$ in $H$, one through $x$ and one through $y$, and we can omit the first arc from one of these paths from $H$ to obtain a graph in which there is a directed path from $r$ to $v$ for every $r \in R$, and a directed path from $v$ to $t$ for every $t \in T$ (note that the arc cannot be part of any path in $H$ starting in $v$, since it is not reachable from $v$ in $H$). This would contradict the minimality of $H$. If $x \in (X_{y}\cup \phi^{-1}(y))$, then there must be an $r \in \phi^{-1}(y)$ such that $x$ is on the path from $r$ to $y$ contradicting the choice of $y$ as $\phi(r)$ and similarly if $y \in (X_{x}\cup \phi^{-1}(x))$. Hence, the sets are indeed disjoint.

Now let us show that $H[R' \cup Y \cup T]$ contains a path from each vertex $r' \in R'$ to each vertex $t \in R' \cup T$. Note that no vertex in $(R \cup S \cup T) \setminus (R' \cup Y \cup T)= (X \cup R) \setminus R'$ is reachable from $v$ in $H$ by the definition of $\phi$ (and $R'$). Each vertex $r' \in (R' \cup T)$ is reachable from $v$ in $H$. Since $\phi(r)$ is a vertex on a path from $r$ to $v$ in $H$, there is a path from $r'$ to $v$ in $H$ for every $r' \in R'$. Since all the vertices on these paths are reachable from $v$, they are not in $(X \cup R) \setminus R'$ and the paths are preserved in $H[R' \cup Y \cup T]$. It follows that $H[R' \cup Y \cup T]$ has the desired property.

We know that $\{X_{r'} \mid r' \in R'\} \cup \{R'\setminus (R \cup T),Y\}$ is a partition of $S$. We also know that $|X_{r'}|$ is an upper bound for $STR(r', \phi^{-1}(r'))$ for each $r' \in R'$ and that $|Y|$ is an upper bound for the minimum size of a solution for the $|R'|$-RST-P instance $(G,R',T \setminus R')$. It follows that the value computed by the algorithm is at most the optimal value for the instance, finishing the proof of correctness.

For the running time, there are $n^{q}$ possible $\phi$'s and for each of them we need to solve a $q'$-RST-P instance for some $q' \le q$ which takes $O(2^{2q+4|T|}\cdot n^{2q+O(1)})$ time and compute $q'$ values of $STR()$. This can be again done using the Nederlof's algorithm (on the reversed digraph), each of the computations taking $O(2^q \cdot n^{O(1)})$ time. Hence, in total, the algorithm runs in $O(2^{2q+4|T|}\cdot n^{3q+O(1)})$ time.
\end{proof}

Suppose now that we restrict to solutions such that there is a vertex $v$ that is an all but a constant number $h$ of the required paths. More formally, for all but $h$ (given) pairs $(r,t) \in R \times T$ the solution contains a path from $r$ to $v$ and at the same time a path from $v$ to $t$. To solve such a variant it is enough to replace the collection of Steiner in-trees used in the proof of Theorem~\ref{thm:trunk} by a solution to a DSN instance, where the set of required pairs is formed by the paths not going over $v$ and the paths from $r$ to $\phi(r)$ for each $r$. Therefore, using the DSN algorithm of Feldman and Ruhl, we can solve also this variant of the problem in $O(2^{O(q+|T|)}\cdot n^{O(q+h)})$ time.

\section{Planar Graphs}\label{sec:planar}
In this section we show how to modify the method of Chitnis et al.~\cite{ChitnisHM14} for SCSS in planar graphs (and graphs excluding a fixed minor) to show the following result.

\begin{theorem}
\label{thm:planar}
 \textsc{$q$-Root Steiner Tree with Pedestal} in planar graphs and graphs excluding a fixed minor can be solved in $O(2^{O(q \log q+|T|\log q)}\cdot n^{O(\sqrt{q})})$ time.
\end{theorem}

The rest of the section is devoted to the proof of this theorem. 
Our goal is to apply the framework of Chitnis et al. to our accelerated game instead of the original game of Feldman and Ruhl. 

The main idea of Chitnis et al. is to introduce supermoves (for more details see Sections 3 and 4 of \cite{ChitnisHM14}) and then show that there is an optimal move sequence of the original game which can be partitioned into $O(q)$ such supermoves. Moreover, there are only $O(q)$ vertices describing the endpoints of the moves. The algorithm then guesses the sequence of the supermoves, without guessing the actual endpoints of the moves, but merely which endpoints are the same (Subsection 5.1 of~\cite{ChitnisHM14}). 

Based on the guess, the algorithm constructs a special directed graph $D$ with vertices being the endpoints of the moves and an extra vertex per each supermove of one of the types. Chitnis et al. then show that the underlying undirected graph of $D$ must be a minor of the underlying undirected graph of the input digraph, or the sequence of supermoves cannot be realized in the input digraph. It follows that $D$ is planar and, as it has $O(q)$ vertices, the treewidth of its underlying undirected graph is $O(\sqrt{q})$~\cite{DemaineH05}. The treewidth of the constructed digraph is checked in $2^{O(\sqrt{q})} \cdot k$ time by the $5$-approximation algorithm of Bodlaender et al.~\cite{BodlaenderDDFLP13}. If the treewidth is too large, the guess is discarded (Subsection 5.2 of~\cite{ChitnisHM14}).

Now they take a universe $U=V \cup V^4$ and construct two functions $cv: V(D) \times U \to \Z^+ \cup \{\infty\}$ and $ce: V(D)^2 \times U^2 \to \Z^+ \cup \{\infty\}$ which describe the costs of mapping the endpoints of particular supermove represented by a vertex or an arc in $D$ to particular vertices in the input digraph. Due to Klein and Marx~\cite{KleinM12}, a mapping $\phi:V(D) \to U$ minimizing $B_{\phi}=\sum_{v \in V(D)} cv(v, \phi(v))+\sum_{\{u,v\} \in E(D)}ce(u,v,\phi(u),\phi(v))$ can be found in $|U|^{O(tw(D))}=n^{O(\sqrt{q})}$ time. The minimum value of such $B_\phi$ over all valid guesses of the algorithm then gives the cost of the optimal solution to the instance (Section 6 of~\cite{ChitnisHM14}).

There are three differences between the original game of Feldman and Ruhl and our accelerated game: 
\begin{enumerate}[(i)]
 \item moves of our game also manipulate the $D$-tokens;
 \item there is no analogue of the move of type (4-b) in the original game; and
 \item there is no move of type (5) in the original game.
\end{enumerate}
The point (iii) is the easiest to solve. Since such a move can be only used at the end of the move sequence, we split the sequence to this move and the rest. We search for the rest of the sequence using the framework by Chitnis et al. and increase the final value by the cost of the type (5) move. In yet another words, we guess the $D$-tokens processed by the move of type (5), remove the corresponding vertices from $T$ and add the cost of the move to the final value.

To solve point (i) we observe that the manipulated $D$-tokens can be seen as an additional label on the move that influences its cost, but does not change its type. When grouping the moves into supermoves, the set of $D$-tokens processed by a supermove is just the union of the sets of $D$-tokens processed by the individual moves.

Finally let us focus on the point (ii). Let $(F,B,D) \xrightarrow{c} ((F\setminus F') \cup \{v\},(B\setminus B') \cup \{f\},D \setminus (D' \cup \{f\}))$, be some move of type (4-b). There is an $F$-token moved from $f$ to $v$ over some path $P$ and a $B$-token moved from some point $b$ on path $P$ to $f$.
Let us suppose that $b$ is the furthest from $f$ among the vertices from $B'$ on $P$.
Ignoring for the moment the $D$-tokens, we could split the move into two parts, the first one being a type (4-a) move between $f$ and $b$ and the other being sequence of (3-a) moves moving the $F$-token from $b$ to $v$.

However, there might be also some $D$-tokens picked up on the part of the path $P$ between $b$ and $v$. For that reason we now have to allow the (3-a) moves to also pick up $D$-tokens (as (3-b) and (4-a) do). Nevertheless, for a type (3-a) move to be able to pick up $D$-tokens, we must be sure that the involved $F$-token already took part in at least one flip, i.e., move of type (4-a). We call such $F$-tokens \emph{active}.
Then the $D$-tokens processed in the (3-a) move can take the path traversed by the $F$-token since the flip in reverse direction to merge with the $B$-token involved in the flip. This condition on the $F$-token being active would complicate the description of the game, since we would have to keep track of the $F$-tokens that are active. However, it is easy to check this condition on the description of the sequence of supermoves as guessed in the Chitnis et al. framework.

Hence, we can use just the moves (3-a), (3-b), and (4-a), where each of them can be equipped with a set of $D$-tokens to be processed within the move. These naturally correspond to single forward, single backward, and flipping moves of the Feldman and Ruhl game. 

We classify the moves the same way as Chitnis et al. into EmptyFlip, NonEmptyFlip, SingleForward, SingleForwardMeet, etc., with one additional parameter being the set of $D$-tokens processed. Note specifically, we call a flip empty even if it processes some $D$-tokens. We define the supermoves Alone, Absorb, AbsorbAndMeet, Meet, NonEmptyFlip, and MultipleFlip and the cleanness of MultipleFlip exactly the same way as Chitnis et al., again equipping each supermove with the set of the processed $D$-tokens.
Now since the set of moves we have is the same as the one of Feldman and Ruhl, the Theorem 5.1 of \cite{ChitnisHM14} applies to our game. That is, there is an optimal sequence of moves for our game that can be partitioned into $O(q)$ supermoves such that each MultipleFlip is clean. Moreover the sets of internal vertices of any two supermoves of this sequence have empty intersection (see Lemma 5.1 of \cite{ChitnisHM14}).

By Lemma 5.2 of \cite{ChitnisHM14}, there are $2^{O(q \log q)}$ unlabeled descriptions for the Feldman and Ruhl game, that is sequences of $O(q)$ supermoves with variables instead of the endpoints. In our case we further describe in which supermove we process which $D$-tokens. Since there are $O(q)$ supermoves and at most $|T|$ $D$-tokens, there are $(O(q))^{|T|}=2^{O(|T|\log q)}$ possible assignments of the $D$-tokens to the supermoves.
Hence, we have $2^{O(q \log q+|T|\log q)}$ unlabeled descriptions for our game in total.

We call an unlabeled description valid if it satisfies all the conditions of Definition 5.1 of \cite{ChitnisHM14} and moreover, for each forward supermove equipped with a non-empty set of $D$-tokens, the involved $F$-token is also involved in a flip supermove prior to this supermove. As in Chitnis et al., the validity of an unlabeled description can be checked in $O(q)$ time.

We associate with a valid unlabeled description $X$ the directed graph $D_X$ exactly the same way Chitnis et al. do. By Theorem 5.2 of \cite{ChitnisHM14}, if the unlabeled description $X^*$ is constructed from the optimal sequence of moves given above, then the underlying undirected graph of $D_{X^*}$ is a minor of the underlying undirected graph of the input digraph $G$. Since $D_X$ has $O(q)$ vertices, due to Demaine and Hajiaghayi~\cite{DemaineH05}, the treewidth of the underlying undirected graph of $D_{X^*}$ is $O(\sqrt{q})$. We can check for each valid unlabeled description $X$ whether the treewidth of the underlying undirected graph of $D_{X}$ is $O(\sqrt{q})$ in $2^{O(\sqrt{q})} \cdot q$ time using the 5-approximation algorithm of Bodlaender et al.~\cite{BodlaenderDDFLP13}.

Before we start guessing the labeling for the unlabeled description, we further guess for each $D$-token processed in a NonEmptyFlip, in which part of the sequence $\alpha \rightarrow \gamma_1 \rightarrow \gamma_2 \cdots \rightarrow \gamma_{\ell} \rightarrow \alpha'$ it is processed. Since there are $O(q)$ possibilities for each of the at most $|T|$ tokens, there are $2^{O(|T|\log q)}$ possibilities in total. We label each edge added to the $D_X$ for the NonEmptyFlip with the set of $D$-tokens processed in that part of the sequence. If we added a new vertex to $D_X$ for a MultipleFlip, then we label it with the the set of $D$-tokens processed in it, otherwise we label the edge added for that move. For all the other supermoves we label the edge added to $D_X$ for the supermove with the the set of $D$-tokens processed in it.

Before we introduce the functions $ce$ and $cv$, let us show how to compute a cost of a MultipleFlip$(f,b,u,v,u',v',D')$, where $D'$ is a set of $D$-tokens to be processed in it. We show that it can be computed similarly to the optimum value of $2$-RST-P instance. We construct the game graph of the accelerated game for the $2$-RST-P instance $(G,\{u,v\}, D')$. Observe that if the two tokens in this game are $f$ and $b$, then all the allowed moves in this accelerated game are (after the decomposition of the (4-b) moves) valid parts of a MultipleFlip. Since both $(\{u\},\{v\}, D')$ and $(\{u'\},\{v'\}, \emptyset)$ are valid configurations of this game, we can compute the cost of the supermove as the length of a shortest path between the two configurations. This takes $O(2^{4|T|}\cdot n^{O(1)})$ time (see Subsection~\ref{sec:algo}).

Now we are ready to define the cost functions $ce$ and $cv$ over the universe $U=V \cup V^4$. We define the function $cv$ exactly the same way as Chitnis et al., taking the appropriate cost of the MultipleFlip with the set $D'$ of processed $D$-tokens as assigned to the node of $D_X$ where relevant. For $ce$ we also take the new cost of MultipleFlip where relevant. Additionally, in the case where Chitnis et al. compute $ce(\alpha, \alpha',v,w)$ from the length of a shortest path from vertex $v$ to vertex $w$ in $G$ for some edge $(\alpha, \alpha')$ of $D_X$ equipped with set $D'$ we replace the formula by $ce(\alpha, \alpha',v,w) = ST(v,\{w\} \cup D')\cdot cv(\alpha,v)\cdot cv(\alpha',w)$, where $ST(x,Y)$ is as defined in Subsection~\ref{sec:accelerated}.

Now by result of Klein and Marx~\cite{KleinM12} we can find a mapping $\phi:V(D) \to U$ minimizing $B_{\phi}=\sum_{v \in V(D)} cv(v, \phi(v))+\sum_{\{u,v\} \in E(D)}ce(u,v,\phi(u),\phi(v))$ in $|U|^{O(tw(D))}=n^{O(\sqrt{q})}$ time.

Our algorithm guesses all possible unlabeled descriptions of appropriate length and for each of them checks its validity. For the valid ones it constructs the directed graph $D_X$ and checks the treewidth of its underlying undirected graph. If it is too large, the guess is discarded, otherwise we further guess the distribution of the $D$-tokens among the NonEmptyFlips. Then we compute the functions $cv$ and $ce$ and find the $\phi$ minimizing the $B_{\phi}$. We output the minimum over all $B_{\phi}$'s obtained this way.

Since there are $2^{O(q \log q+|T|\log q)}$ unlabeled descriptions and $2^{O(|T|\log q)}$ distributions of the $D$-tokens, there are $2^{O(q \log q+|T|\log q)}$ possible different guesses in total. Checking validity and constructing the digraph takes linear time, checking the treewidth takes $2^{O(\sqrt{q})}$ time. There are $O(n^4)$ vertices in the universe $U$ and $O(n^8)$ their pairs. For each of them, the computation of $cv$ or $ce$, respectively, takes at most $O(2^{4|T|}\cdot n^{O(1)})$ time. Finally, the execution of the algorithm of Klein and Marx takes $n^{O(\sqrt{q})}$ time.
Therefore, the algorithm runs in $2^{O(q \log q+|T|\log q)} \cdot n^{O(\sqrt{q})}$ time in total.

The correctness follows from the representability of the accelerated game and from Lemmata 6.2 and 6.3 of \cite{ChitnisHM14}. This finishes the proof of the theorem.

\section{Conclusion and Future Directions}\label{sec:concl}
We have shown that there is a nice special case of DSN that allows for as effective algorithms as were known for DST and SCSS, even with respect to planar graphs. We characterized that the crucial property of the solution to allow this is the existence of a vertex over which almost all paths required by the problem definition ``factorize''. An interesting open question is what is the complexity of $q$-RST (the unrestricted variant) in planar graphs.

Another interesting question is tied to the other parameterization of the problems. We are not aware of any result determining the complexity of SCSS in planar graphs with respect to the parameterization the number of nonterminals in the solution.

\bibliographystyle{splncs03}
\bibliography{roots}

\end{document}